\documentclass{article}
\usepackage{amsfonts,amsmath,amsthm,amscd,amssymb,latexsym}
\usepackage{tikz}
\DeclareMathOperator{\tr}{tr}
\theoremstyle{plain}
\newtheorem{thm}{Theorem}[section]
\newtheorem{cor}[thm]{Corollary}

\newtheorem{prop}[thm]{Proposition}
\newtheorem{rem}[thm]{Remark}
\theoremstyle{definition}
\theoremstyle{remark}
\numberwithin{equation}{section}
\newcommand{\keywords}{\textbf{Key words and phrases: }\medskip}
\newcommand{\subjclass}{\textbf{Math. Subj. Clas.: }\medskip}

\begin{document}
\title{\textbf{2D  discrete Yang-Mills equations on the torus} }
\author{\textbf{Volodymyr Sushch} \\
{ \em Koszalin University of Technology} \\
 {\em Sniadeckich 2, 75-453 Koszalin, Poland} \\
 { \em volodymyr.sushch@tu.koszalin.pl} }

\date{}
\maketitle
\begin{abstract}
In this paper, we introduce a discretization scheme for the Yang-Mills equations in the two-dimensional case using a framework based on discrete exterior calculus. Within this framework, we define discrete versions of the exterior covariant derivative operator and its adjoint, which capture essential geometric features similar to their continuous counterparts. Our focus is on discrete models defined on a combinatorial torus, where the discrete Yang-Mills equations are presented in the form of both a system of difference equations and a matrix form.
\end{abstract}



 \section{Introduction}

 Discrete models of mathematical physics problems that preserve a geometric structure are crucial for the successful numerical simulation of differential equations.
 In \cite{S6}, a discrete version of the de Rham-Hodge theory in the specific case of a two-dimensional torus has been discussed in detail. In this paper, motivated by this and considering the Yang-Mills theory as nonlinear generalization of Hodge theory, we construct a discrete counterpart of Yang-Mills equations in the two-dimensional case. Our discretization scheme is based on the geometric approach proposed by Dezin in \cite{Dezin}.
 This approach is characterized by the construction of  a discrete exterior calculus on a complex of real-valued cochains. In author's previous papers \cite{S1, S2, S3, S4, S5}, this technique has been developed on  matrix-valued complexes of cochains, and discrete analogues of the Yang-Mills equations in $\mathbb{R}^n$ and in Minkowski space have been studied. In \cite{S1, S3}, we  focused on the construction of gauge invariant discrete models. Depending on the definition of a discrete version of the Hodge star, we discuss two types of discrete Yang-Mills equations. In \cite{S4, S5}, we study discrete analogues of instanton and anti-instanton solutions.

There are various approaches to discretizing Yang-Mills theories, with numerous papers dedicated to this subject. For instance, see \cite{Catterall1, Catterall2, Joseph, KSSMKI, HK} and the references therein. Many of these approaches rely on lattice discretization schemes. However, in the lattice formulation, maintaining the geometrical properties of the original gauge theory poses challenges. Two dimensional Yang-Mills theories have been investigated from several points of view in \cite{AB,  AKS, DGHK, GKS, Witten}.

Recently, from a numerical simulation perspective, interest in the discretization of the Yang-Mills equations has been fueled by the development of methods that preserve the geometric and topological structures of their continuum counterparts. Constraint-preserving numerical methods for the Yang-Mills equations, based on the finite element exterior calculus \cite{Arnold}, have been developed and investigated in \cite{CW, CH, BKS}. In \cite{DOQ}, within framework of a discrete de Rham method,  another structure-preserving numerical scheme to approximate the Yang-Mills equations has been constructed.

Let's provide brief reminders of the
main definitions of smooth Yang-Mills
theory.
Although our focus in this paper is on the two dimensional torus,
we will describe the Yang–Mills action in the more general context of compact Riemannian
manifolds of arbitrary dimension.
Let $M$ be a compact connected Riemannian manifold. For simplicity, we  assume that a principal $G$-bundle $P$ over $M$ is trivial and that the structure  group $G$ is the Lie group $SU(2)$. A connection on $P$ is defined as the $su(2)$-valued differential 1-form $A$, where $su(2)$ is the Lie algebra of the Lie group $SU(2)$. The curvature 2-form $F$ of the connection $A$ is defined as follows
\begin{equation}\label{1.1}
 F=dA+A\wedge A,
 \end{equation}
where $d$ is the exterior derivative and $\wedge$ is the exterior product.
The exterior covariant derivative $d_A$ of the connection $A$ acts on an $su(2)$-valued  $r$-form $\Omega$ as
 \begin{equation}\label{1.2}
d_A\Omega=d\Omega+A\wedge\Omega+(-1)^{r+1}\Omega\wedge A.
\end{equation}
Denote by $\delta_A$ the formal adjoint operator of $d_A$ with respect to the $L^2$-norm on the sections of $P$. It is given by
$\delta_A=\pm\ast d_A\ast$, where $\ast$ is the Hodge star. See  \cite{AB} for details.
The Yang-Mills equations in terms of differential forms can be represented as
\begin{equation}\label{1.3}
  \delta_A F=0,
 \end{equation}
or, equivalently, as
\begin{equation}\label{1.4}
 d_A\ast F=0.
 \end{equation}
 Equations \eqref{1.3} or \eqref{1.4} form a system of second order nonlinear differential equations. A solution $A$ of \eqref{1.3} or \eqref{1.4} is referred to as the Yang-Mills connection. Furthermore, by virtue of the Bianchi identity $d_A F=0$,  the Yang-Mills connection satisfies the following Laplace-type equation
 \begin{equation}\label{1.5}
  (d_A\delta_A+\delta_A d_A) F=0.
 \end{equation}
 This explains why the Yang-Mills theory can be interpreted as nonlinear generalization of  Hodge theory. See, for example, \cite{FU} or \cite{AB} for details.

 Our purpose is to propose a geometric structure-preserving discrete model of the Yang-Mills equations on a two-dimensional torus.
 The technique developed in \cite{S6} to describe a discrete analogue of the de Rham-Hodge theory also allows us to construct and investigate discrete analogues of Equations \eqref{1.3} and \eqref{1.4} on a combinatorial torus.
 We show that discrete analogues of the operators $d_A$ and $\delta_A$ exhibit properties similar to those in the continuous case. Furthermore, we present a system of difference equations explicitly, which are equivalent to discrete Yang-Mills equations.

The paper is organized as follows. In Section 2, we briefly review the concept of a combinatorial model of $\mathbb{R}^2$, revisiting relevant material from \cite{S6}. We introduce a discrete version of matrix-valued differential forms (a complex of $su(2)$-valued cochains) and define analogues of the main exterior calculus operations on them. Section 3 establishes a discrete analogue of the Yang-Mills equations in the two-dimensional case. Additionally, we discuss the properties of discrete counterparts of the covariant exterior derivative and its adjoint operator. In Section 4, we apply the proposed discretization scheme to a combinatorial torus. We present the discrete Yang-Mills equations in terms of a system of difference equations and also provide the matrix form of this system.

  \section{Background on a discrete model}
  The construction of a combinatorial model of the two-dimensional Euclidean space
 $\mathbb{R}^2$ is presented in detail in \cite{S6}. We will only recall the main definitions and key operations that are most relevant to this paper.
  A combinatorial model of  $\mathbb{R}^2$
   consists of  a two-dimensional chain complex $C(2)=C_0(2)\oplus C_1(2)\oplus C_2(2)$ generated by the 0-, 1-, and 2-dimensional basis elements $\{x_{k,s}\}$,
$\{e_{k,s}^1,  e_{k,s}^2\}$, and $\{V_{k,s}\}$, respectively, where $k,  s \in {\mathbb Z}$.  The boundary operator $\partial: C_r(2)\rightarrow C_{r-1}(2)$
 is defined by
\begin{align}\label{2.1}
\partial x_{k,s}=0, \qquad  \partial e_{k,s}^1=x_{\tau k,s}-x_{k,s} \qquad  \partial e_{k,s}^2=x_{k, \tau s}-x_{k,s},\nonumber \\
\partial V_{k,s}=e_{k,s}^1+e_{\tau k,s}^2-e_{k, \tau s}^1-e_{k,s}^2,
\end{align}
where $\tau$ denotes the shift by one to the right, i.e., $\tau k=k+1$.
The definition \eqref{2.1} is extended to arbitrary chains by linearity.

Let us consider an object dual to the chain complex $C(2)$. The dual complex $K(2)=K^0(2)\oplus K^1(2)\oplus K^2(2)$ is defined as in \cite{S6}. It has a  structure similar to $C(2)$ and consists of cochains  with coefficients belonging to the algebra $su(2)$. For the construction of
$su(2)$-valued cochains see also \cite{S1}. Denote by $\{x^{k,s}\}$,  \,
$\{e^{k,s}_1, \ e^{k,s}_2\}$, and $\{V^{k,s}\}$ the basis elements of $K^0(2)$, $K^1(2)$, and $K^2(2)$, respectively. Then the cochains
 $\Phi\in K^0(2)$,  $\Omega\in K^1(2)$, and $\Psi\in K^2(2)$ can be written as

  \begin{equation}\label{2.2}
  \Phi=\sum_{k,s}\Phi_{k,s}x^{k,s}, \quad
  \Omega=\sum_{k,s}(\Omega^1_{k,s}e_1^{k,s}+\Omega^2_{k,s}e_2^{k,s}), \quad \Psi=\sum_{k,s}\Psi_{k,s}V^{k,s},
\end{equation}
where  $\Phi_{k,s}, \Omega^1_{k,s}, \Omega^2_{k,s}, \Psi_{k,s}\in su(2)$ for any $k, s \in {\mathbb Z}$.
 As in \cite{S6}, we will call cochains forms (discrete forms).
 Note the difference between $K(2)$ here and the cochain complex in \cite{S6}. The complex $K(2)$ consists of $su(2)$-valued cochains while in \cite{S6} cochains are real-valued.

For discrete forms \eqref{2.2}, the pairing is defined with the basis elements of $C(2)$ according to the rule
\begin{equation}\label{2.3}
\langle x_{k,s}, \ \Phi\rangle=\Phi_{k,s}, \ \langle e_{k,s}^1, \ \Omega\rangle=\Omega^1_{k,s}, \
\langle e_{k,s}^2, \  \Omega\rangle=\Omega^2_{k,s},  \ \langle V_{k,s}, \Psi \rangle=\Psi_{k,s}.
\end{equation}

As in \cite{S6}, we consider the coboundary operator $d^c: K^r(2)\rightarrow K^{r+1}(2)$.
The coboundary operator can be interpreted as a discrete analogue of the exterior differential. For the forms \eqref{2.2}, we have
\begin{equation}\label{2.4}
d^c\Phi=\sum_{k,s}(\Phi_{\tau k,s}-\Phi_{k,s})e_1^{k,s}+(\Phi_{k,\tau s}-\Phi_{k,s})e_2^{k,s},
\end{equation}
\begin{equation}\label{2.5}
d^c\Omega=\sum_{k,s}(\Omega_{\tau k,s}^2-\Omega_{k,s}^2-\Omega_{k,\tau s}^1+\Omega_{k,s}^1)V^{k,s},
\end{equation}
and $d^c\Psi=0$.

 Finally, let's recall the definitions of the $\cup$ product and the star operator on $K(2)$, following \cite{S6}.
 For the basis elements of $K(2)$, the $\cup$ product is defined as follows
\begin{equation*}\label{}
x^{k,s}\cup x^{k,s}=x^{k,s}, \qquad x^{k,s}\cup e^{k,s}_1=e^{k,s}_1, \qquad x^{k,s}\cup e^{k,s}_2=e^{k,s}_2,
\end{equation*}
\begin{equation*}\label{}
x^{k,s}\cup V^{k,s}=V^{k,s}, \qquad e^{k,s}_1\cup x^{\tau k,s}=e^{k,s}_1, \qquad e^{k,s}_2\cup x^{k, \tau s}=e^{k,s}_2,
\end{equation*}
\begin{equation*}\label{}
V^{k,s}\cup x^{\tau k,\tau s}=V^{k,s}, \qquad e^{k,s}_1\cup e^{\tau k,s}_2=V^{k,s}, \qquad e^{k,s}_2\cup e^{k, \tau s}_1=-V^{k,s},
\end{equation*}
with the product being zero in all other cases. This operation is extended to arbitrary forms by linearity, where the form coefficients are multiplied as matrices. In \cite[Ch.3, Proposition 2]{Dezin},  it is shown that for real-valued discrete forms, the discrete counterpart of the Leibniz rule  is valid. For any matrix-valued forms, this rule holds as well, and we have
\begin{equation}\label{2.6}
 d^c(\Omega\cup\Phi)=d^c\Omega\cup\Phi+(-1)^r\Omega\cup
d^c\Phi,
\end{equation}
where $r$ is the degree of $\Omega$.
The star operator $\ast: K^r(2)\rightarrow  K^{2-r}(2)$ is defined  by the rule
\begin{equation}\label{2.7}
\ast x^{k,s}=V^{k,s}, \quad \ast e^{k,s}_1=e^{\tau k,s}_2, \quad \ast e^{k,s}_2=-e^{k,\tau s}_1, \quad \ast V^{k,s}=x^{\tau k, \tau s}.
\end{equation}
Again, it's  extended to arbitrary forms by linearity. The operation $\ast$ has properties similar to the Hodge star operator.  Thus, it can be regarded as a discrete analogue of the Hodge star operator.

Let's now shift our focus to the inner product on $K(2)$.
We'll start by recalling that the real Lie algebra $su(2)$ has a basis given by
 \begin{equation*}
E_1=
\begin{bmatrix}
i & 0  \\
0 & -i
\end{bmatrix}, \quad
E_2=
\begin{bmatrix}
0 & 1  \\
-1 & 0
\end{bmatrix}, \quad
E_3=
\begin{bmatrix}
0 & i  \\
i & 0
\end{bmatrix},
\end{equation*}
where $i$ is the imaginary  unit.
Any matrix $A\in su(2)$ can be expressed as
\begin{equation}\label{2.8}
A=a_1E_1+a_2E_2+a_3E_3=
\begin{bmatrix}
a_1i & a_2+a_3i  \\
-a_2+a_3i & -a_1i
\end{bmatrix},
\end{equation}
 where $a_1, a_2, a_3\in\mathbb{R}$.
 The algebra $su(2)$ is endowed with an inner product defined as
 \begin{equation*}
(A,  B)=-\tr(AB).
 \end{equation*}
 Here, $\tr$ denotes the trace of the matrix product.
 It's easy to show that this forms a real symmetric bilinear form, and in fact, it's positive definite.

Let's denote by $V$ the two-dimensional finite chain with unit coefficients, defined by
\begin{equation}\label{2.9}
  V=\sum_{k=1}^N\sum_{s=1}^MV_{k,s}.
\end{equation}
Now, we define the inner product of discrete forms on $V$ as follows
 \begin{equation}\label{2.10}
 (\Phi, \ \Omega)_V=-\frac{1}{2}\tr\langle V, \ \Phi\cup\ast\Omega\rangle,
 \end{equation}
 where  $\Phi$ and $\Omega$ are discrete forms of the same degree.
   For  forms of different degrees, the product \eqref{2.10} is set equal to zero.
  From \eqref{2.3} and \eqref{2.7}, using the definition of the $\cup$ product, we obtain
   \begin{equation}\label{2.11}
 (\Phi, \ \Omega)_V=-\frac{1}{2}\tr\sum_{k=1}^N\sum_{s=1}^M\Phi_{k,s}\Omega_{k,s}=\sum_{k=1}^N\sum_{s=1}^M\sum_{\alpha=1}^3\varphi_{k,s,\alpha}\omega_{k,s,\alpha},
 \end{equation}
where $\Phi_{k,s}=[\varphi_{k,s,\alpha}]$ and $\Omega_{k,s}=[\omega_{k,s,\alpha}]$ are the matrices of the form \eqref{2.8}. Thus, the norm is given by
\begin{equation}\label{2.12}
 \|\Phi\|^2=(\Phi, \ \Phi)_V=\sum_{k=1}^N\sum_{s=1}^M\sum_{\alpha=1}^3(\varphi_{k,s,\alpha})^2.
 \end{equation}
Now, let's define the adjoint operator of $d^c$ with respect to the inner product \eqref{2.10}. As in \cite[Proposition 2]{S6},  we denote the adjoint operator of $d^c$  as $\delta^c$ which satisfies the following relation.

\begin{prop}
 Let $\Phi\in K^r(2)$  and $\Omega\in K^{r+1}(2)$,  $r=0,1$. Then we have
\begin{equation}\label{2.13}
 (d^c\Phi, \ \Omega)_V=-\frac{1}{2}\tr\langle \partial V, \ \Phi\cup\ast\Omega\rangle+(\Phi, \ \delta^c\Omega)_V,
\end{equation}
 where
 \begin{equation}\label{2.14}
 \delta^c\Omega=(-1)^{r+1}\ast^{-1}d^c\ast\Omega
 \end{equation}
  and $\ast^{-1}$ is the inverse of $\ast$.
\end{prop}
\begin{proof}
Taking into account \eqref{2.6} and \eqref{2.10}, the proof coincides with that \cite[Proposition 2]{S6}.
\end{proof}
It is clear that the operator $\delta^c: K^{r+1}(2) \rightarrow K^r(2)$ defined by \eqref{2.14} serves as a  discrete counterpart to the codifferential $\delta$ (the operator adjoint of $d$).
Note that by \eqref{2.7}, we have
\begin{equation*}\label{}
\ast^{-1} x^{k,s}=V^{\sigma k,\sigma s}, \quad \ast^{-1} e^{k,s}_1=-e^{k,\sigma s}_2, \quad \ast^{-1} e^{k,s}_2=e^{\sigma k, s}_1, \quad \ast^{-1} V^{k,s}=x^{k, s}.
\end{equation*}
Here, $\sigma$  denotes the shift by one to the left, i.e., $\sigma k=k-1$. Combining these relations with \eqref{2.4} and \eqref{2.5}
 for the forms \eqref{2.2}, we obtain $\delta^c\Phi=0$, and
\begin{equation}\label{2.15}
\delta^c\Omega=\sum_{k=1}^N\sum_{s=1}^M(\Omega_{\sigma k,s}^{1}-\Omega_{k,s}^{1}+\Omega_{k,\sigma s}^{2}-\Omega_{k,s}^{2})x^{k,s},
\end{equation}
\begin{equation}\label{2.16}
\delta^c\Psi=\sum_{k=1}^N\sum_{s=1}^M(\Psi_{k,s}-\Psi_{k,\sigma s})e_1^{k,s}
-(\Psi_{k,s}-\Psi_{\sigma k,s})e_2^{k,s}.
\end{equation}

\begin{rem}
The inner product is defined over $V$ in the form \eqref{2.9}, but the relation \eqref{2.13} includes terms with the components
$\Phi_{0,s}$, \ $\Phi_{\tau N,s}$, \ $\Omega_{0,s}$, $\Omega_{\tau N,s}$, \
$\Phi_{k,0}$, \ $\Phi_{k, \tau M}$, $\Omega_{k,0}$  and $\Omega_{k, \tau M}$, which are not specified by the definition of the form $\Phi$, and $\Omega$ on $V$. Hence, these components must be defined additionally.
\end{rem}
Assume the components of the form \eqref{2.2} satisfy the following conditions
\begin{align}\label{2.17}
\Phi_{0,s}=\Phi_{N,s}, \quad \Phi_{\tau N,s}=\Phi_{1,s},  \quad \Phi_{k,0}=\Phi_{k,M}, \quad \Phi_{k,\tau M}=\Phi_{k,1}, \nonumber \\
 \Psi_{0,s}=\Psi_{N,s}, \quad \Psi_{\tau N,s}=\Psi_{1,s},  \quad \Psi_{k,0}=\Psi_{k,M}, \quad \Psi_{k,\tau M}=\Psi_{k,1}, \nonumber \\
\Omega^1_{0,s}=\Omega^1_{N,s}, \quad \Omega^1_{\tau N,s}=\Omega^1_{1,s},  \quad \Omega^1_{k,0}=\Omega^1_{k,M}, \quad \Omega^1_{k,\tau M}=\Omega^1_{k,1},\nonumber \\
\Omega^2_{0,s}=\Omega^2_{N,s}, \quad \Omega^2_{\tau N,s}=\Omega^2_{1,s},   \quad \Omega^2_{k,0}=\Omega^2_{k,M}, \quad \Omega^2_{k,\tau M}=\Omega^2_{k,1}
\end{align}
for any $k=1,2, ..., N$ and $s=1,2,.., M$.
We denote by $H^r(V)$ the Hilbert space generated by the inner product \eqref{2.10} of $r$-forms satisfying conditions \eqref{2.17}.
For the operators
\begin{equation*}\label{}
d^c: H^r(V) \rightarrow H^{r+1}(V),  \qquad \delta^c: H^{r+1}(V) \rightarrow H^r(V),
\end{equation*}
 the relation \eqref{2.13} takes the form
\begin{equation}\label{2.18}
 (d^c\Phi, \ \Omega)_V=(\Phi, \ \delta^c\Omega)_V.
\end{equation}
We'll omit the proof here. See  \cite[Proposition 3]{S6} for details. However, it's worth noting that
\begin{equation*}\label{}
  \partial V=\sum_{k=1}^Ne^1_{k,1}+\sum_{s=1}^Me^2_{\tau N,s}-\sum_{k=1}^Ne^1_{k,\tau M}-\sum_{s=1}^Me^2_{1,s},
\end{equation*}
and for any $\Phi\in H^r(V)$, $\Omega\in H^{r+1}(V)$, we have
\begin{equation*}
\langle \partial V, \ \Phi\cup\ast\Omega\rangle=0.
\end{equation*}
The  operator
\begin{equation*}\label{}
\Delta^c=d^c\delta^c+\delta^cd^c: H^r(V) \rightarrow H^r(V)
\end{equation*}
  serves as a discrete analogue of the Laplacian.
Furthermore,   $\Delta^c$  is self-adjoint, i.e.,
\begin{equation*}\label{}
 (\Delta^c\Phi, \ \Omega)_V=(\Phi, \ \Delta^c\Omega)_V,
\end{equation*}
where $\Phi$, $\Omega\in H^r(V)$. This follows immediately from \eqref{2.18}.

  \section{Discrete Yang-Mills equations}

In this section, we establish a discrete counterpart of the Yang-Mills equations for the model considered in the previous section.
We also discuss some properties of operators generated by the discrete Yang-Mills equations.
Let the 1-form
\begin{equation}\label{3.1}
    A=\sum_{k=1}^N\sum_{s=1}^M(A^1_{k,s}e_1^{k,s}+A^2_{k,s}e_2^{k,s}),
\end{equation}
where $A^1_{k,s}, A^2_{k,s}\in su(2)$,
be a discrete analogue of the $su(2)$-valued connection 1-form. The discrete analogue of the curvature form \eqref{1.1} is defined by
\begin{equation}\label{3.2}
    F=d^cA+A\cup A.
\end{equation}
Using \eqref{2.5} and the definition of $\cup$, the components $F_{k,s}\in su(2)$ of the 2-form
\begin{equation*}\label{}F=\sum_{k=1}^N\sum_{s=1}^MF_{k,s}V^{k,s}
\end{equation*}
 can be represented as
\begin{equation}\label{3.3}
    F_{k,s}=A^2_{\tau k,s}-A^2_{k,s}-A^1_{k,\tau s}+A^1_{k,s}+A^1_{k,s}A^2_{\tau k,s}-A^2_{k,s}A^1_{k,\tau s}
\end{equation}
for any $k, s$.
Now, let's consider the operator  $d^c_A: H^r(V) \rightarrow H^{r+1}(V)$ defined by
 \begin{equation}\label{3.4}
d^c_A\Omega=d^c\Omega+A\cup\Omega+(-1)^{r+1}\Omega\cup A,
\end{equation}
where $A\in H^1(V)$ and $\Omega\in H^r(V)$. Since $A$ represents the discrete connection 1-form, we refer to $d^c_A$ as the discrete analogue of the covariant exterior derivative \eqref{1.2}.
It is obvious that the  equation  $d^c_AF=0$ is an identity (a  discrete analogue of the Bianchi identity).
The discrete counterpart of the Yang-Mills equations \eqref{1.4} can be given as
\begin{equation}\label{3.5}
  d^c_A\ast F=0,
 \end{equation}
 where $F$ is the discrete curvature form \eqref{3.2}.
  The form $\ast F$ is a 0-form, and applying \eqref{2.7} yields
 \begin{equation*}\label{}
 \ast F=\sum_{k=1}^N\sum_{s=1}^MF_{k,s}x^{\tau k,\tau s}=\sum_{k=1}^N\sum_{s=1}^MF_{\sigma k,\sigma s}x^{k,s},
 \end{equation*}
 where $F_{k,s}$ is given by \eqref{3.3}. Then Equation  \eqref{3.5} can be represented as
\begin{equation*}\label{}
 \langle e^1_{k,s}, \  d_A^c\ast F\rangle=F_{k,\sigma s}-F_{\sigma k,\sigma s}+A^1_{k,s}F_{k,\sigma s}-F_{\sigma k,\sigma s}A^1_{k,s}=0,
 \end{equation*}
 \begin{equation*}\label{}
 \langle e^2_{k,s}, \  d_A^c\ast F\rangle=F_{\sigma k, s}-F_{\sigma k,\sigma s}+A^2_{k,s}F_{\sigma k, s}-F_{\sigma k,\sigma s}A^2_{k,s}=0
 \end{equation*}
 for any $k=1,2, ..., N$ and $s=1,2, ..., M$.

 The operator $d_A^c$ behaves like the coboundary operator $d^c$ under the $\cup$ operation. It follows in particular that we have the following:
 \begin{prop}
 Let $\Omega\in H^r(V)$. Then we have
 \begin{equation}\label{3.6}
 d^c_A(\Omega\cup\Phi)=d^c_A\Omega\cup\Phi+(-1)^r\Omega\cup
d^c_A\Phi.
\end{equation}
\end{prop}

\begin{proof}
Let $\Omega\in H^r(V)$  and $\Phi\in H^p(V)$, where $r=0, p=1$ or $r=1, p=0$. Then  by \eqref{2.6} and \eqref{3.4}, we obtain
\begin{align*}
d_A^c(\Omega\cup\Phi)=d^c(\Omega\cup\Phi)+A\cup\Omega\cup\Phi+(-1)^{r+p+1}\Omega\cup\Phi\cup A \\
=d^c\Omega\cup\Phi+(-1)^r\Omega\cup d^c\Phi+A\cup\Omega\cup\Phi+(-1)^{r+p+1}\Omega\cup\Phi\cup A \\
=d^c\Omega\cup\Phi+A\cup\Omega\cup\Phi+(-1)^{r+1}\Omega\cup A\cup\Phi\\
+(-1)^r(\Omega\cup d^c\Phi+\Omega\cup A\cup\Phi+(-1)^{p+1}\Omega\cup\Phi\cup A) \\
=d^c_A\Omega\cup\Phi+(-1)^r\Omega\cup d^c_A\Phi.
\end{align*}
  In cases where $\Omega\cup\Phi\in H^2(V)$, both sides of the relation \eqref{3.6} are zero.
  Therefore, Equation \eqref{3.6} holds for all cases, as desired.
\end{proof}
It should be noted that under conditions \eqref{2.17}, the following 2-forms
\begin{equation*}
\Phi=\sum_{k=1}^N\sum_{s=1}^M\Phi_{k,s}V^{k,s}, \ \tilde{\Phi}=\sum_{k=1}^N\sum_{s=1}^M\Phi_{\tau k,s}V^{k,s}, \
\hat{\Phi}=\sum_{k=1}^N\sum_{s=1}^M\Phi_{k,\tau s}V^{k,s}
\end{equation*}
take the same value over $V$, i.e.,
\begin{equation}\label{3.7}
\langle V, \ \Phi\rangle=\langle V, \ \tilde{\Phi}\rangle=\langle V, \ \hat{\Phi}\rangle=\sum_{k=1}^N\sum_{s=1}^M\Phi_{k,s}.
\end{equation}

\begin{prop}
Let $\Phi\in H^r(V)$ and $\Omega\in H^{2-r}(V)$.
Then we have
\begin{equation}\label{3.8}
 \tr\langle V, \ \Phi\cup\Omega\rangle=\tr\langle V, \ \Omega\cup\ast\ast\Phi\rangle.
\end{equation}
\end{prop}
\begin{proof}
Let's start with forms $\Phi\in H^2(V)$ and $\Omega\in H^0(V)$. These forms are
\begin{equation*}
\Phi=\sum_{k=1}^N\sum_{s=1}^M\Phi_{k,s}V^{k,s}, \ \Omega=\sum_{k=1}^N\sum_{s=1}^M\Omega_{k,s}x^{k,s}.
\end{equation*}
By the definition of the $\cup$ product, we have
\begin{equation*}
\Phi\cup\Omega=\sum_{k=1}^N\sum_{s=1}^M\Phi_{k,s}\Omega_{\tau k,\tau s}V^{k,s}
\end{equation*}
and thus,
\begin{equation*}
\tr\langle V, \ \Phi\cup\Omega\rangle=\sum_{k=1}^N\sum_{s=1}^M\tr(\Phi_{k, s}\Omega_{\tau k,\tau s}).
\end{equation*}
Since
\begin{equation*}
\ast\ast\Phi=\sum_{k=1}^N\sum_{s=1}^M\Phi_{k,s}V^{\tau k,\tau s}=\sum_{k=1}^N\sum_{s=1}^M\Phi_{\sigma k, \sigma s}V^{k,s},
\end{equation*}
we calculate
\begin{equation*}
\Omega\cup\ast\ast\Phi=\sum_{k=1}^N\sum_{s=1}^M\Omega_{k,s}\Phi_{\sigma k, \sigma s}V^{k,s}.
\end{equation*}
By \eqref{2.17}, it is true that for any 2-form we have
\begin{equation*}\label{}
 \langle V, \ \Phi\rangle=\sum_{k=1}^N\sum_{s=1}^M\Phi_{k,s}=\sum_{k=1}^N\sum_{s=1}^M\Phi_{\sigma k, \sigma s}=\langle V, \ \ast\ast\Phi\rangle.
\end{equation*}
Hence,
\begin{align*}
\tr\langle V, \  \Omega\cup\ast\ast\Phi\rangle=\sum_{k=1}^N\sum_{s=1}^M\tr(\Omega_{k, s}\Phi_{\sigma k, \sigma s})=
\sum_{k=1}^N\sum_{s=1}^M\tr(\Omega_{\tau k,\tau s}\Phi_{k, s})\\
=\sum_{k=1}^N\sum_{s=1}^M\tr(\Phi_{k, s}\Omega_{\tau k,\tau s})=\tr\langle V, \ \Phi\cup\Omega\rangle.
\end{align*}
Here, we use the fact that $\tr(AB)=\tr(BA)$ for any two matrices $A$ and $B$ of appropriate sizes.

Similarly, we derive \eqref{3.8} for $\Phi\in H^0(V)$ and $\Omega\in H^2(V)$.

Finally, let's consider the forms
 $\Phi\in H^1(V)$ and $\Omega\in H^1(V)$ given by
\begin{equation*}
\Phi=\sum_{k=1}^N\sum_{s=1}^M(\Phi^1_{k,s}e_1^{k,s}+\Phi^2_{k,s}e_2^{k,s}), \ \Omega=\sum_{k=1}^N\sum_{s=1}^M(\Omega^1_{k,s}e_1^{k,s}+\Omega^2_{k,s}e_2^{k,s}).
\end{equation*}
In this case, the $\cup$ product is
\begin{equation*}
\Phi\cup\Omega=\sum_{k=1}^N\sum_{s=1}^M(\Phi^1_{k,s}\Omega^2_{\tau k,s}-\Phi^2_{k,s}\Omega^1_{k,\tau s})V^{k,s}.
\end{equation*}
This gives
\begin{equation*}
\tr\langle V,\Phi\cup\Omega\rangle=\sum_{k=1}^N\sum_{s=1}^M\tr(\Phi^1_{k,s}\Omega^2_{\tau k,s}-\Phi^2_{k,s}\Omega^1_{k,\tau s}).
\end{equation*}
Using \eqref{2.7}, we get
\begin{equation}\label{3.9}
\ast\ast\Phi=-\sum_{k=1}^N\sum_{s=1}^M(\Phi^1_{\sigma k,\sigma s}e_1^{k,s}+\Phi^2_{\sigma k,\sigma s}e_2^{k,s}).
\end{equation}
Then we have
\begin{equation*}\label{}
\Omega\cup\ast\ast\Phi=\sum_{k=1}^N\sum_{s=1}^M(\Omega^2_{k,s}\Phi^1_{\sigma k, s}-\Omega^1_{k,s}\Phi^2_{k,\sigma s})V^{k,s}.
\end{equation*}
By \eqref{3.7}, we obtain
\begin{align*}
\tr\langle V, \ \Omega\cup\ast\ast\Phi\rangle=\sum_{k=1}^N\sum_{s=1}^M\tr(\Omega^2_{k,s}\Phi^1_{\sigma k, s}-\Omega^1_{k,s}\Phi^2_{k,\sigma s})\\
=\sum_{k=1}^N\sum_{s=1}^M\tr(\Omega^2_{\tau k,s}\Phi^1_{k, s}-\Omega^1_{k,\tau s}\Phi^2_{k, s})
=\sum_{k=1}^N\sum_{s=1}^M\tr(\Phi^1_{k, s}\Omega^2_{\tau k,s}-\Phi^2_{k, s}\Omega^1_{k,\tau s})\\
=\tr\langle V, \ \Phi\cup\Omega\rangle.
\end{align*}
\end{proof}
\begin{prop}
Let $\Phi\in H^r(V)$ and $\Omega\in H^{2-r}(V)$.
Then
\begin{equation*}\label{}
 (\Phi, \ \ast^{-1}\Omega)_V=(\ast\Phi, \ \Omega)_V.
 \end{equation*}
 \end{prop}
\begin{proof}
From the definition of the inner product \eqref{2.10} and by  Proposition~3.2,  it follows  that
\begin{align*}
(\Phi, \ \ast^{-1}\Omega)_V=-\frac{1}{2}\tr\langle V, \  \Phi\cup\ast\ast^{-1}\Omega\rangle=-\frac{1}{2}\tr\langle V, \  \Phi\cup\Omega\rangle \\
=-\frac{1}{2}\tr\langle V, \  \Omega\cup\ast\ast\Phi\rangle=(\Omega, \ \ast\Phi)_V=(\ast\Phi, \ \Omega)_V.
\end{align*}
\end{proof}
\begin{cor}
\begin{equation*}\label{}
 \tr\langle V, \ \Phi\cup\Omega\rangle=\tr\langle V, \ \ast\Phi\cup\ast\Omega\rangle.
\end{equation*}
\end{cor}

\begin{prop}
 Let $\Phi\in H^r(V)$  and $\Omega\in H^{r+1}(V)$,  $r=0,1$. Then we have
\begin{equation}\label{3.10}
 (d_A^c\Phi, \ \Omega)_V=(\Phi, \ \delta_A^c\Omega)_V,
\end{equation}
where
\begin{equation}\label{3.11}
  \delta_A^c\Omega=\delta^c\Omega+\ast^{-1}(\ast\Omega\cup\ast\ast A+(-1)^{r+1}A\cup\ast\Omega).
\end{equation}
\end{prop}
\begin{proof}
This is a fairly straight forward proof. Using \eqref{2.18}, \eqref{3.4}, and \eqref{3.8}, we have
\begin{align*}
(d_A^c\Phi, \ \Omega)_V=-\frac{1}{2}\tr\langle V, \  d_A^c\Phi\cup\ast\Omega\rangle\\
=-\frac{1}{2}\tr\langle V, \  (d^c\Phi+A\cup\Phi+(-1)^{r+1}\Phi\cup A)\cup\ast\Omega\rangle\\
=(d^c\Phi, \ \Omega)_V-\frac{1}{2}\tr\langle V, \ A\cup\Phi\cup\ast\Omega\rangle-(-1)^{r+1}\frac{1}{2}\tr\langle V, \  \Phi\cup A\cup\ast\Omega\rangle\\
=(\Phi, \ \delta^c\Omega)_V-\frac{1}{2}\tr\langle V, \ \Phi\cup\ast\Omega\cup\ast\ast A\rangle-(-1)^{r+1}\frac{1}{2}\tr\langle V, \  \Phi\cup A\cup\ast\Omega\rangle\\
=(\Phi, \ \delta^c\Omega)_V-\frac{1}{2}\tr\langle V, \ \Phi\cup\ast\ast^{-1}(\ast\Omega\cup\ast\ast A+(-1)^{r+1} A\cup\ast\Omega)\rangle\\
=(\Phi, \ \delta^c\Omega)_V+(\Phi, \ \ast^{-1}(\ast\Omega\cup\ast\ast A))_V+(-1)^{r+1}(\Phi, \ \ast^{-1}(A\cup\ast\Omega))_V\\
=(\Phi, \ \delta^c\Omega+\ast^{-1}(\ast\Omega\cup\ast\ast A+(-1)^{r+1}A\cup\ast\Omega))_V
=(\Phi, \ \delta_A^c\Omega)_V.
\end{align*}
\end{proof}
Let us consider the following operator
\begin{equation*}\label{}
\Delta_A^c=d_A^c\delta_A^c+\delta_A^cd_A^c: H^r(V) \rightarrow H^r(V).
\end{equation*}
This operator serves as a discrete analogue of the Laplace-type operator \eqref{1.5}.

\begin{prop}
 For any  $r$-form $\Phi\in H^r(V)$,   we have
\begin{equation*}\label{}
 (\Delta_A^c\Phi, \ \Phi)_V\geq 0
\end{equation*}
and $(\Delta_A^c\Phi, \ \Phi)_V=0$ if and only if $\Delta_A^c\Phi=0$.
\end{prop}
\begin{proof}
By \eqref{3.10}, one has
\begin{align*}\label{}
 (\Delta_A^c\Phi, \ \Phi)_V=(d_A^c\delta_A^c\Phi, \ \Phi)_V+(\delta_A^cd_A^c\Phi, \ \Phi)_V\\=
 (\delta_A^c\Phi, \ \delta_A^c\Phi)_V+(d_A^c\Phi, \ d_A^c\Phi)_V=\|\delta_A^c\Phi\|^2+\|d_A^c\Phi\|^2,
\end{align*}
where $\|\cdot\|$ is the norm given by \eqref{2.12}.
From this, if $(\Delta_A^c\Phi, \ \Phi)_V=0$, then $\|\delta_A^c\Phi\|^2=0$ and $\|d_A^c\Phi\|^2=0$.  It gives $\delta_A^c\Phi=0$ and $d_A^c\Phi=0$. Hence,
\begin{equation*}
\Delta_A^c\Phi=d^c_A\delta_A^c\Phi+\delta_A^cd_A^c\Phi=0.
\end{equation*}
\end{proof}
\begin{cor}
$\Delta_A^c\Phi=0$ if and only if $\delta_A^c\Phi=0$ and $d_A^c\Phi=0$.
\end{cor}
\begin{prop}
 The operator  $\Delta_A^c: H^r(V) \rightarrow H^r(V)$  is self-adjoint, i.e.,
\begin{equation*}\label{}
 (\Delta_A^c\Phi, \ \Omega)_V=(\Phi, \ \Delta_A^c\Omega)_V.
\end{equation*}
\end{prop}
\begin{proof}
By \eqref{3.10}, it is obvious.
\end{proof}
\begin{rem}
Taking into account  \eqref{2.14}, the operator $\delta_A^c\Phi$ can be  expressed  as
\begin{equation}\label{3.12}
  \delta_A^c\Phi=(-1)^r\ast^{-1}(d^c\ast\Phi+(-1)^r\ast\Phi\cup\ast\ast A+A\cup\ast\Phi),
\end{equation}
where $\Phi\in H^r(V)$. On the other hand, by \eqref{3.4}, it is straightforward to calculate that
\begin{equation}\label{3.13}
  \delta_A^c\Phi=(-1)^r\ast^{-1}d_A^c\ast\Phi+\ast^{-1}(\ast\Phi\cup(\ast\ast A+A)).
\end{equation}
It's worth noting that in the continuous case, we have $\delta_A=(-1)^r\ast^{-1}d_A\ast$. However, this doesn't hold true in our discrete model. More precisely, it follows from \eqref{2.7} that the operation $\ast\ast$ corresponds to a shift, meaning $\ast\ast A$ takes the form \eqref{3.9}, unlike in the continuous case where $\ast\ast A=-A$. Consequently, the discrete Yang-Mills equations \eqref{3.5} does not follow  immediately from the equation  $\Delta_A^c F=0$.
\end{rem}
The discrete counterpart of Equation \eqref{1.3} has the form
\begin{equation}\label{3.14}
  \delta_A^c F=0,
\end{equation}
where $F$ is given by \eqref{3.2}. Again, using \eqref{2.16} and \eqref{3.9}, from \eqref{3.11}, Equation \eqref{3.14} can be represented in  terms of difference  equations as follows
\begin{equation*}\label{}
 F_{k,\sigma s}-F_{\sigma k,\sigma s}+A^1_{k,s}F_{k,\sigma s}-F_{\sigma k,\sigma s}A^1_{\sigma k,\sigma s}=0,
 \end{equation*}
 \begin{equation*}\label{}
 F_{\sigma k, s}-F_{\sigma k,\sigma s}+A^2_{k,s}F_{\sigma k, s}-F_{\sigma k,\sigma s}A^2_{\sigma k,\sigma s}=0
 \end{equation*}
 for any $k=1,2, ..., N$ and $s=1,2, ..., M$. Comparing the difference representations of Equations \eqref{3.5} and \eqref{3.14}, it can be noted that only the last terms of the difference equations slightly differ in both cases.

\section{Discrete Yang-Mills equations on a combinatorial torus}

In this section, following \cite{S6}, we delve into an example of our discrete model on a combinatorial torus in detail. Recall that the torus can be regarded as the topological space obtained by taking a rectangle and identifying each pair of opposite sides with the same orientation. As was already discussed  in \cite{S6}, let's consider the partitioning of the plane $\mathbb{R}^2$   by the straight lines
$x=k$ and $y=s$, where $k,s\in\mathbb{Z}$. We denote an open square bounded by these lines as
$V_{k,s}$, with vertices labeled as $x_{k,s}, \ x_{\tau k,s}, \ x_{k,\tau s}$, \ $x_{\tau k, \tau s}$, where $\tau k=k+1$. Further, we define
$e_{k,s}^1$ and $e_{k,s}^2$ as the open intervals  $(x_{k,s}, \ x_{\tau k,s})$ and $(x_{k,s}, \ x_{k, \tau s})$, respectively.
These geometric objects can be associated with the combinatorial objects we have previously discussed. We identify the collection
$V_{k,s}$ with $V$  given by  \eqref{2.9}  and let
$N=M=2$.
If we identify the points and the intervals on the boundary of
$V$  as follows
\begin{align}\label{4.1}
x_{1,1}=x_{3,1}=x_{1,3}=x_{3,3}, \qquad x_{1,2}=x_{3,2}, \qquad x_{2,1}=x_{2,3}, \nonumber \\
e_{1,1}^1=e_{1,3}^1, \qquad e_{2,1}^1=e_{2,3}^1, \qquad e_{1,1}^2=e_{3,1}^2, \qquad e_{1,2}^2=e_{3,2}^2,
\end{align}
we obtain the geometric object which is homeomorphic to the torus. For a visual representation, see \cite[Figure~1]{S6}. Let $C(T)$ denote the complex corresponding to this geometric object, serving as a combinatorial model of the torus. Similarly,  we denote by $K(T)$ the complex of cochains over $C(T)$.
From \eqref{4.1}, it follows that forms  on the combinatorial torus satisfy the conditions \eqref{2.17}.

Let us now proceed  to the construction of the discrete Yang-Mills equations on the combinatorial torus. The discrete connection 1-form   $A\in K^1(T)$ has the form
\begin{equation*}\label{}
    A=\sum_{k=1}^2\sum_{s=1}^2(A^1_{k,s}e_1^{k,s}+A^2_{k,s}e_2^{k,s}).
\end{equation*}
Using the definition of the $\cup$ product and \eqref{2.5}, we have
\begin{align*}\label{}
 A\cup A=(A^1_{1,1}A^2_{2,1}-A^2_{1,1}A^1_{1,2})V^{1,1}+(A^1_{2,1}A^2_{1,1}-A^2_{2,1}A^1_{2,2})V^{2,1}\nonumber\\
  +(A^1_{1,2}A^2_{2,2}-A^2_{1,2}A^1_{1,1})V^{1,2}+(A^1_{2,2}A^2_{1,2}-A^2_{2,2}A^1_{2,1})V^{2,2}
\end{align*}
and
\begin{align*}\label{}
 d^cA=(A^1_{1,1}-A^1_{1,2}+A^2_{2,1}-A^2_{1,1})V^{1,1}+(A^1_{2,1}-A^1_{2,2}-A^2_{2,1}+A^2_{1,1})V^{2,1}\nonumber\\
  +(A^1_{1,2}-A^1_{1,1}+A^2_{2,2}-A^2_{1,2})V^{1,2}+(A^1_{2,2}-A^1_{2,1}+A^2_{1,2}-A^2_{2,2})V^{2,2}.
\end{align*}
Then the discrete curvature form  \eqref{3.2} on $C(T)$ can be represented as
\begin{equation*}\label{}
     F=F_{1,1}V^{1,1}+F_{2,1}V^{2,1}+F_{1,2}V^{1,2}+F_{2,2}V^{2,2},
\end{equation*}
where
\begin{align*}\label{}
 F_{1,1}=A^1_{1,1}-A^1_{1,2}+A^2_{2,1}-A^2_{1,1}+A^1_{1,1}A^2_{2,1}-A^2_{1,1}A^1_{1,2},\\
 F_{2,1}=A^1_{2,1}-A^1_{2,2}-A^2_{2,1}+A^2_{1,1}+A^1_{2,1}A^2_{1,1}-A^2_{2,1}A^1_{2,2},\\
 F_{1,2}=A^1_{1,2}-A^1_{1,1}+A^2_{2,2}-A^2_{1,2}+A^1_{1,2}A^2_{2,2}-A^2_{1,2}A^1_{1,1}, \\
 F_{2,2}=A^1_{2,2}-A^1_{2,1}+A^2_{1,2}-A^2_{2,2}+A^1_{2,2}A^2_{1,2}-A^2_{2,2}A^1_{2,1}.
\end{align*}
By \eqref{2.7} and \eqref{4.1}, we obtain
\begin{equation*}\label{}
     \ast F=F_{2,2}x^{1,1}+F_{1,2}x^{2,1}+F_{2,1}x^{1,2}+F_{1,1}x^{2,2}.
\end{equation*}
Consequently, we can calculate the following:
\begin{align}\label{4.2}
 d^c\ast F=(F_{1,2}-F_{2,2})e_1^{1,1}+(F_{2,2}-F_{1,2})e_1^{2,1}+(F_{2,2}-F_{2,1})e_2^{1,2}\nonumber\\+(F_{2,1}-F_{2,2})e_2^{1,1}+(F_{1,1}-F_{2,1})e_1^{1,2}+(F_{2,1}-F_{1,1})e_1^{2,2}+
  \nonumber\\+(F_{1,2}-F_{1,1})e_2^{2,2}+(F_{1,1}-F_{1,2})e_2^{2,1},
\end{align}
\begin{align}\label{4.3}
\ast F\cup A=F_{2,2}A^1_{1,1}e_1^{1,1}+F_{1,2}A^1_{2,1}e_1^{2,1}+F_{2,1}A^2_{1,2}e_2^{1,2}+
  F_{2,2}A^2_{1,1}e_2^{1,1}\nonumber \\+F_{2,1}A^1_{1,2}e_1^{1,2}+F_{1,1}A^1_{2,2}e_1^{2,2}+F_{1,1}A^2_{2,2}e_2^{2,2}+F_{1,2}A^2_{2,1}e_2^{2,1},
\end{align}
and
\begin{align}\label{4.4}
A\cup\ast F=A^1_{1,1}F_{1,2}e_1^{1,1}+A^1_{2,1}F_{2,2}e_1^{2,1}+A^2_{1,2}F_{2,2}e_2^{1,2}+ A^2_{1,1}F_{2,1}e_2^{1,1}\nonumber \\+
 A^1_{1,2}F_{1,1}e_1^{1,2}+A^1_{2,2}F_{2,1}e_1^{2,2}+A^2_{2,2}F_{1,2}e_2^{2,2}+A^2_{2,1}F_{1,1}e_2^{2,1}.
\end{align}
Combining \eqref{4.2}, \eqref{4.3}, and \eqref{4.4}, we obtain the representation of $d_A^c\ast F$ on the combinatorial torus.
Then Equation \eqref{3.5} can be expressed in terms of difference equations as
\begin{align*}\label{}
 \langle e^1_{1,1}, \  d_A^c\ast F\rangle=F_{1,2}-F_{2,2}+A^1_{1,1}F_{1,2}-F_{2,2}A^1_{1,1}=0,\\
 \langle e^1_{2,1}, \  d_A^c\ast F\rangle=F_{2,2}-F_{1,2}+A^1_{2,1}F_{2,2}-F_{1,2}A^1_{2,1}=0,\\
 \langle e^2_{1,2}, \  d_A^c\ast F\rangle=F_{2,2}-F_{2,1}+A^2_{1,2}F_{2,2}-F_{2,1}A^2_{1,2}=0,\\
 \langle e^2_{1,1}, \  d_A^c\ast F\rangle=F_{2,1}-F_{2,2}+A^2_{1,1}F_{2,1}-F_{2,2}A^2_{1,1}=0,\\
 \langle e^1_{1,2}, \  d_A^c\ast F\rangle=F_{1,1}-F_{2,1}+A^1_{1,2}F_{1,1}-F_{2,1}A^1_{1,2}=0,\\
 \langle e^1_{2,2}, \  d_A^c\ast F\rangle=F_{2,1}-F_{1,1}+A^1_{2,2}F_{2,1}-F_{1,1}A^1_{2,2}=0,\\
 \langle e^2_{2,2}, \  d_A^c\ast F\rangle=F_{1,2}-F_{1,1}+A^2_{2,2}F_{1,2}-F_{1,1}A^2_{2,2}=0,\\
 \langle e^2_{2,1}, \  d_A^c\ast F\rangle=F_{1,1}-F_{1,2}+A^2_{2,1}F_{1,1}-F_{1,2}A^2_{2,1}=0.
\end{align*}
Similarly, taking into account \eqref{3.13},  Equation \eqref{3.14} is equivalent to the following system of difference equations
\begin{align}\label{4.5}
 F_{1,2}-F_{2,2}+A^1_{1,1}F_{1,2}-F_{2,2}A^1_{2,2}=0, \nonumber \\
 F_{2,2}-F_{1,2}+A^1_{2,1}F_{2,2}-F_{1,2}A^1_{1,2}=0, \nonumber \\
 F_{2,2}-F_{2,1}+A^2_{1,2}F_{2,2}-F_{2,1}A^2_{2,1}=0, \nonumber \\
 F_{2,1}-F_{2,2}+A^2_{1,1}F_{2,1}-F_{2,2}A^2_{2,2}=0,\nonumber\\
 F_{1,1}-F_{2,1}+A^1_{1,2}F_{1,1}-F_{2,1}A^1_{2,1}=0,\nonumber\\
 F_{2,1}-F_{1,1}+A^1_{2,2}F_{2,1}-F_{1,1}A^1_{1,1}=0,\nonumber\\
 F_{1,2}-F_{1,1}+A^2_{2,2}F_{1,2}-F_{1,1}A^2_{1,1}=0,\nonumber\\
 F_{1,1}-F_{1,2}+A^2_{2,1}F_{1,1}-F_{1,2}A^2_{1,2}=0.
\end{align}

Finally, we give a matrix form of Equations \eqref{3.5} and \eqref{3.14} on the combinatorial torus. Let's introduce the following row vectors
\begin{equation*}\label{}
 [F]=[F_{1,1} \ F_{2,1} \ F_{1,2} \ F_{2,2} ], \quad
 [A]=[A^1_{1,1} \ A^1_{2,1} \ A^2_{1,2} \ A^2_{1,1} \ A^1_{1,2} \ A^1_{2,2} \ A^2_{2,2} \ A^2_{2,1}],
 \end{equation*}
 \begin{equation*}\label{}
 [x]=[x^{1,1} \ x^{2,1} \ x^{1,2} \ x^{2,2} ], \quad
 [e]=[e_1^{1,1} \ e_1^{2,1} \ e_2^{1,2} \ e_2^{1,1} \ e_1^{1,2} \ e_1^{2,2} \ e_2^{2,2} \ e_2^{2,1}],
 \end{equation*}
 \begin{equation*}\label{}
  [V]=[V^{1,1} \ V^{2,1} \ V^{1,2} \ V^{2,2}].
 \end{equation*}
 Denote by $[\cdot]^T$ the corresponding column vector. In this notation, the 1-form $A$ can be rewritten as
 \begin{equation*}\label{}
 A=[e][A]^T=[A][e]^T.
 \end{equation*}
 By trivial computation, one finds that
 \begin{equation*}\label{}
 [\ast F]=[F]S, \quad \ast F=[F]S[x]^T
   \end{equation*}
 and
 \begin{equation*}\label{}
  d^c\ast F=[e]DS[F]^T, \quad
  \quad \ast F\cup A=[F]SD_2I_A[e]^T, \quad A \cup\ast F=[e]I_AD_1S[F]^T,
 \end{equation*}
 where
\begin{equation*}
D=
\begin{bmatrix}
-1 & 1 & 0 & 0 \\
1 & -1 & 0 & 0 \\
1 & 0 & -1 & 0 \\
-1 & 0 & 1 & 0 \\
0 & 0 & -1 & 1 \\
0 & 0 & 1 & -1 \\
0 & 1 & 0 & -1 \\
0 & -1 & 0 & 1
\end{bmatrix}, \quad
S=
\begin{bmatrix}
0 & 0 & 0 & 1 \\
0 & 0 & 1 & 0 \\
0 & 1 & 0 & 0 \\
1 & 0 & 0 & 0
\end{bmatrix},
\end{equation*}
\begin{equation*}
D_1=
\begin{bmatrix}
0 & 1 & 0 & 0 \\
1 & 0 & 0 & 0 \\
1 & 0 & 0 & 0 \\
0 & 0 & 1 & 0 \\
0 & 0 & 0 & 1 \\
0 & 0 & 1 & 0 \\
0 & 1 & 0 & 0 \\
0 & 0 & 0 & 1
\end{bmatrix}, \quad
D_2=
\begin{bmatrix}
1 & 0 & 0 & 1 & 0 & 0 & 0 & 0 \\
0 & 1 & 0 & 0 & 0 & 0 & 0 & 1\\
0 & 0 & 1 & 0 & 1 & 0 & 0 & 0\\
0 & 0 & 0 & 0 & 0 & 1 & 1 & 0
\end{bmatrix},
\end{equation*}
\begin{equation*}
I_A=
\begin{bmatrix}
A^1_{1,1} & 0 & 0 & 0 & 0 & 0 & 0 & 0\\
0 & A^1_{2,1} & 0 & 0 & 0 & 0 & 0 & 0\\
0 & 0 & A^2_{1,2} & 0 & 0 & 0 & 0 & 0\\
0 & 0 & 0 & A^2_{1,1} & 0 & 0 & 0 & 0 \\
0 & 0 & 0 & 0 & A^1_{1,2} & 0 & 0 & 0\\
0 & 0 & 0 & 0 & 0 & A^1_{2,2} & 0 & 0 \\
0 & 0 & 0 & 0 & 0 & 0 & A^2_{2,2} & 0 \\
0 & 0 & 0 & 0& 0 & 0 & 0 & A^2_{2,1}
\end{bmatrix}.
\end{equation*}
Thus, the discrete Yang-Mills equation \eqref{3.5} can be written in  matrix form as
\begin{equation*}\label{}
 [e]DS[F]^T+[e]I_AD_1S[F]^T-[F]SD_2I_A[e]^T=0
 \end{equation*}
or, equivalently,
 \begin{equation*}\label{}
 DS[F]^T+I_AD_1S[F]^T-([F]SD_2I_A)^T=[0]^T.
 \end{equation*}
Similarly, Equation \eqref{3.14} or the system \eqref{4.5} can be represented as follows
 \begin{equation*}\label{}
 DS[F]^T+I_AD_1S[F]^T+([F]SD_2I_{\ast\ast A})^T=[0]^T,
 \end{equation*}
where
\begin{equation*}\label{}
 [\ast\ast A]=[-A^1_{2,2} \ -A^1_{1,2} \ -A^2_{2,1} \ -A^2_{2,2} \ -A^1_{2,1} \ -A^1_{1,1} \ -A^2_{1,1} \ -A^2_{1,2}].
 \end{equation*}

 \section{Conclusions}
This study introduces a discrete model of the two-dimensional Yang-Mills equations that preserves geometric structure. The properties of operators generated by the discrete Yang-Mills equations have been thoroughly investigated. Special attention is given to discrete models on the combinatorial two-dimensional torus. In this context, a system of difference equations equivalent to the discrete Yang-Mills equations was constructed, and the matrix form of these equations was also presented. The author encourages readers to apply the proposed discretization scheme to develop numerical implementation techniques.

 \end{document}